\documentclass[runningheads]{llncs}
\usepackage{graphicx}
\usepackage{latexsym}
\usepackage{amsmath, amssymb, colonequals}
\usepackage{graphics, graphicx, url, color, setspace, lineno, enumerate, verbatim}
\usepackage{tikz}
\usepackage{pdflscape}
\usepackage{longtable}
\usepackage{geometry}
\usepackage[english]{babel}
\usepackage{setspace}
\usepackage{bm}
\usepackage{hyperref}
\usepackage{comment}
\usepackage{float}
\usepackage{caption}
\usepackage{subcaption}
\usepackage{algorithmic}
\usepackage{algorithm}
\usepackage{pbox}
\usepackage[titletoc,toc,page]{appendix}
\usepackage{multirow}							%to add rows in the table that spans multiple rows
\usepackage[misc]{ifsym}						%for the envelope symbol

\usepackage{xr}

\def\C{\mathbb{C}}
\def\Z{\mathbb{Z}}

\def\proj{\mathop{\rm proj}}

\def\S{\mathcal{S}}
\def\LP{\mathrm{LP}}

\def\C{\mathcal{C}}
\def\Cr{\mathrm{C}}

\newcommand{\vc}[1]{\bm{#1}}	%defines vector format
\newcommand{\mc}[1]{\mathcal{#1}}	%defines font format
\begin{document}
\title{Consistency for 0--1 Programming}
\author{Danial Davarnia\inst{1} \and
J.\ N.\ Hooker\inst{2}
}
\authorrunning{Davarnia and Hooker}
\institute{Iowa state University\\
\email{davarnia@iastate.edu}\\ \and
Carnegie Mellon University\\
\email{jh38@andrew.cmu.edu}}
\maketitle              % typeset the header of the contribution
\begin{abstract}
Concepts of consistency have long played a key role in constraint programming but never developed in integer programming (IP). Consistency nonetheless plays a role in IP as well. For example, cutting planes can reduce backtracking by achieving various forms of consistency as well as by tightening the linear programming (LP) relaxation. We introduce a type of consistency that is particularly suited for 0-1 programming and develop the associated theory. We define a 0-1 constraint set as LP-consistent when any partial assignment that is consistent with its linear programming relaxation is consistent with the original 0-1 constraint set. We prove basic properties of LP-consistency, including its relationship with Chv\'atal-Gomory cuts and the integer hull. We show that a weak form of LP-consistency can reduce or eliminate backtracking in a way analogous to $k$-consistency but is easier to achieve. In so doing, we identify a class of valid inequalities that can be more effective than traditional cutting planes at cutting off infeasible 0-1 partial assignments. 
\end{abstract}

\section{Introduction} \label{sec:introduction}

Consistency is a fundamental concept of constraint programming (CP) and an essential tool for the reduction of backtracking during search \cite{Apt03}.  Curiously, the concept never explicitly developed in mathematical programming, even though solvers rely on a similar type of branching search.  In fact, the cutting planes of integer programming can reduce backtracking by achieving various forms of consistency as well as by tightening the linear programming (LP) relaxation.   

This suggests that it may be useful to investigate the potential role of consistency concepts in mathematical programming.  We do so for 0--1 integer programming in particular.  We study how consistency relates to such integer programming ideas as the LP relaxation, Chv\'atal-Gomory cutting planes \cite{Chv73}, and the integer hull, as well as how consistency can be achieved for 0--1 inequalities.  Our main contribution is to introduce a type of consistency, {\em LP-consistency}, that seems particularly relevant to 0--1 programming, and to develop the underlying theory.  We show that achieving a form of partial LP-consistency can reduce backtracking in ways that traditional cutting planes cannot.

One way to reduce backtracking is to identify partial assignments to the variables that are {\em inconsistent} with the constraint set, meaning that they cannot occur in a feasible solution of the constraints.  Branching decisions that result in such partial assignments can then be avoided, thus removing infeasible subtrees from the search.  Unfortunately, it is generally hard to identify inconsistent partial assignments in advance.  

The essence of consistency is that it makes it easier to identify inconsistent partial assignments.  Full consistency allows one to recognize an inconsistent partial assignment by the fact that it violates a constraint that contains only the variables in the partial assignment.  Because full consistency is very hard to achieve, CP solvers rely on {\em domain consistency} (generalized arc consistency) \cite{Apt03,Dav87,Mac77,Mon74}, which reduces variable domains to the point that every value in them occurs in some feasible solution.  If domain consistency is obtained at the current node of the search tree, branching on any value in a variable's domain can lead to a feasible solution.  Domain consistency is itself hard to achieve for the entire constraint set, but can often be achieved, or partially achieved, for individual global constraints in the CP model, and this reduces backtracking significantly \cite{Reg10}.

Our approach is based on the idea that consistency can be defined with respect to a {\em relaxation} of the constraint set.  Specifically, we interpret consistency as making it possible to identify inconsistent partial assignments by checking whether they are consistent with a certain type of relaxation.  This perspective allows us to propose alternative types of consistency by using various types of relaxation.  For traditional consistency, the relaxation is obtained simply by dropping constraints that contain variables that are not in the partial assignment.  We define LP-consistency by replacing this relaxation with the LP relaxation. Thus LP-consistency ensures that any partial assignment that is consistent with the LP relaxation is inconsistent with the original constraint set.  Fortunately, one can easily check consistency with an LP relaxation simply by solving the LP problem that results from adding the partial assignment to the LP relaxation.

This poses the question of whether it is practical to achieve LP-consistency for a 0--1 problem.  There is no known practical method for achieving full LP-consistency, but we take a cue from the concept of $k$-consistency in CP \cite{Fre78,Tsa93,Hen89}, which is weaker than full consistency but sufficient to avoid backtracking if the constraints are not too tightly coupled by common variables.  While achieving traditional $k$-consistency is impractical, we define a similar property, {\em sequential LP $k$-consistency}, that can be computed for small $k$.  This in turn can avoid some backtracking that traditional cutting planes may permit, because it focuses on identifying inconsistent partial assignments rather than cutting off fractional solutions of the LP relaxation.  

A method for obtaining sequential LP $k$-consistency is suggested by our practice of defining all consistency concepts in terms of projection, as proposed in \cite{Hoo16}.   One can define sequential LP $k$-consistency, in particular, in terms of the results of lifting a problem from $k-1$ dimensions to $k$ dimensions, and then projecting it back into $k-1$ dimensions.  This same lift-and-project operation is carried out by a special case of the widely-used lift-and-project technique \cite{Bal85}, and we show that this procedure obtains sequential LP $k$-consistency.  

We begin below by defining and illustrating basic consistency concepts and showing how they can be cast in terms of projection.  We also indicate how consistency can eliminate or reduce backtracking.  We review some prior work showing that an inference method of propositional logic, resolution, can achieve consistency for 0--1 problems, and that a weak form of resolution, input resolution, can generate all Chv\'{a}tal-Gomory cuts for a set of logical clauses.   

At this point we introduce LP-consistency and show some elementary properties, namely that consistency implies LP-consistency, and a constraint set that describes the integer hull is necessarily LP-consistent.  Yet LP-consistency is a concept that does not occur in polyhedral theory, and an LP-consistent constraint set need not describe the integer hull.  While the facet-defining inequalities that describe the integer hull are generally regarded as the strongest valid inequalities, we show that they can be weaker than a non-facet-defining inequality that achieves LP-consistency, in the sense that they exclude fewer inconsistent 0--1 (partial) assignments.  We further elaborate on connections with cutting plane theory by showing that a 0--1 partial assignment is consistent with the LP relaxation if and only if it violates no logical clause that is a Chv\'{a}tal-Gomory (C-G) cut, and a 0--1 problem is LP-consistent if and only if all of its implied logical clauses are C-G cuts.  We also note that while input resolution derives C-G cuts, it does not achieve LP-consistency.

The remainder of the paper defines and develops the concept of sequential LP $k$-consistency.  It shows that achieving sequential LP $k$-consistency for $k=1,\ldots, n$ (where $n$ is the number of variables) avoids backtracking altogether for branching order $x_1, x_2, \dotsc, x_n$. In practice, one would achieve sequential LP $k$-consistency for a few small values of $k$.  We then prove that a restricted version of the well-known lift-and-project procedure \cite{Bal85} achieves sequential $k$-consistency for a given $k$.  Finally, we illustrate how achieving sequential LP $k$-consistency even for $k=2$ can avoid backtracking that is permitted by traditional separating cuts.

\section{Consistency and Projection}

To define consistency, it is convenient to adopt basic terminology as follows.  The {\em domain} $D_j$ of a variable $x_j$ is the set of values that can be assigned to $x_j$.  A {\em constraint} $C$ is an object that {\em contains} some set $\{x_1,\ldots,x_k\}$ of variables, such that any given assignment of values to $(x_1,\ldots,x_k)$ either {\em satisfies} or {\em violates} $C$.  Thus a constraint is satisfied or violated only when all of its variables have been assigned values.  An assignment to $x$ satisfies a constraint set $\S$ when it satisfies all the constraints in $\S$.
A list of symbols defined hereafter appears in Table~\ref{ta:symbols}.
% and violates $\S$ when it violates at least one constraint in $\S$.

\begin{table}[b]
\vspace{-4ex}
{\small
    \centering
    \caption{List of symbols.} \label{ta:symbols}
    \hspace{20ex}
    \begin{tabular}{l@{\hspace{2ex}}l}
    \ \\
    \hline
    $x_J$ & tuple of variables $x_j$ for $j\in J$\\
    $D(\S)$ & satisfaction set of constraint set $\S$ \\
    $D_J(\S)$ & set of assignments to $x_J$ that are consistent with $\S$ \\
    $D_j$ & domain of $x_j$ \\
    $D(\S)|_J$ & projection of $D(\S)$ onto $x_J$ \\
    $\S_J$ & set of constraints in $\S$ that contain only variables in $x_J$ \\
    $S_{\LP}$ & LP relaxation of 0--1 constraint set $\S$ \\
    $D_J(\S_{\LP})$ & set of 0--1 assignments to $x_J$ that are consistent with $\S_{\LP}$ \\
    $\S_{\Cr}$ & set of clausal inequalities implied by individual constraints of $\S$ \\
    $J_k$ & $\{1,\ldots,k\}$ \\
    \hline
    \end{tabular}
}
\end{table}

Let $x_J$ be the tuple containing the variables in $\{x_j\;|\;j\in J\}$ for $J\subseteq N=\{1,\ldots,n\}$.  A {\em partial assignment} to $x$ is an assignment of values to $x_J$ for some $J\subseteq N$.
%, while a {\em complete assignment} assigns values to all variables in $x$.  
We can now define a consistent partial assignment and a consistent constraint set.

\begin{definition}
	Given a constraint set $\S$, a partial assignment $x_J=v_J$ is {\em consistent with $\S$} if $\S\cup\{x_J=v_J\}$ is feasible.  
\end{definition}
Since it is hard in general to determine whether $\S\cup\{x_J=v_J\}$ is feasible, it is hard to identify which partial assignments are consistent with $\S$.  Consistent constraint sets are defined so that it is easy to identify which partial assignments are consistent with them.  

\begin{definition}
A constraint set $\S$ is {\em consistent} if every partial assignment to $x$ that violates no constraint in $\S$ is consistent with $\S$.
\end{definition}
The contrapositive is perhaps more intuitive: $\S$ is consistent when every partial assignment that is inconsistent with $\S$ violates some individual constraint in $\S$.   Thus a consistent constraint set can be viewed as one in which implied constraints are made explicit, in the sense that every inconsistent partial assignment is explicitly ruled out by some constraint in the set.

Since full consistency is generally hard to achieve, the constraint programming community has found various weaker forms of consistency to be more useful.  By far the most popular is domain consistency, also known as generalized arc consistency \cite{Apt03,Dav87,Mac77,Mon74}.  
\begin{definition}
A constraint set $\S$ is {\em domain consistent} if $x_j=v_j$ is consistent with $\S$ for all $v_j\in D_j$ and all variables $x_j$.  
\end{definition}
That is, every value in the domain of a variable $x_j$ is assigned to $x_j$ in some feasible solution of $\S$.  A consistent constraint set is necessarily domain consistent.

\begin{example} \label{ex:0}
Suppose that $\S$ is the constraint set 
\[
\begin{array}{cccccccc}
x_1 & + & x_2 &   &     & + & x_4 & \geq 1 \\
x_1 & - & x_2 & + & x_3 &   &     & \geq 0 \\
x_1 &   &     &   &     & - & x_4 & \geq 0 \\
\multicolumn{7}{l}{x_j\in \{0,1\}, \;\mbox{all}\; j}
\end{array}
\]
The feasible solutions $(x_1,\ldots, x_4)$ of $\S$ are listed below:
\[
\begin{array}{c@{\hspace{5ex}}c@{\hspace{5ex}}c}
(0,1,1,0) & (1,0,1,0) & (1,1,0,1) \\
(1,0,0,0) & (1,0,1,1) & (1,1,1,0) \\
(1,0,0,1) & (1,1,0,0) & (1,1,1,1)
\end{array}
\]
Set $\S$ is not consistent because, for instance, the partial assignment $(x_1,x_2)=(0,0)$ violates no constraint in $\S$ but is inconsistent with $\S$ due to the fact that $(x_1,x_2)=(0,0)$ in none of the feasible solutions.  On the other hand, $\S$ is domain consistent because $x_j=0$ and $x_j=1$ occur in some feasible solution for each $j$.
\end{example}

The various consistency concepts are more easily defined in terms of projection, as proposed in \cite{Hoo16}.
Let $D(\S)$ be the satisfaction set of $\S$; that is, the set of assignments to $x$ that satisfy $\S$.  Also let $D_J(\S)$ be the set of partial assignments $x_j=v_j$ that are consistent with $\S$, so that $D_J(\S)= \{v_J\;|\;\S\cup\{x_J=v_J \}\;\mbox{is feasible}\}$.  The {\em projection} of $D(\S)$ onto $x_J$, which we may write $D(\S)|_J$, is $\{x_J\;|\;x\in D(\S)\}$.  Note that the projection is identical to the set of assignments to $x_J$ that are consistent with $\S$, so that $D(S)|_J=D_J(\S)$.

This last observation allows us to define consistency in terms of projection.  Let $\S_J$ be the set of constraints in $\S$ whose variables belong to $x_J$.  Then $D_J(\S_J)$ is the set of assignments to $x_J$ that violate no constraints in $\S$.  We assume that ${\S}$ contains the {\em in-domain constraints} $x_j\in D_j$ for all $j\in N$.

\begin{proposition}
A constraint set $\S$ is consistent if and only if $D_J(\S_J)=D(\S)|_J$ for all $J\subseteq N$.  Equivalently, $\S$ is consistent if and only if all 0--1 partial assignments $x_J=v_J$ that are consistent with $\S_J$ are consistent with $\S$.  In addition, $\S$ is domain consistent if and only if  $D_j=D(\S)|_{\{j\}}$ for all $j\in N$.
\end{proposition}

\begin{example} \label{ex:00}
If $\S$ is as in Example~\ref{ex:0}, $\S$ is not consistent because, for instance, the satisfaction set $D_{\{1,2\}}(\S_{\{1,2\}})=\{(0,0),(0,1),(1,0),(1,1)\}$ of $\S_{\{1,2\}}=\{x_1\in \{0,1\}, \; x_2\in\{0,1\}\}$ is different from the projection onto $(x_1,x_2)$ of $D(\S)$, which is $D(\S)|_{\{1,2\}}=\{(0,1),(1,0),(1,1)\}$.  However, $\S$ is domain consistent because $D_j = \{0,1\} = D(\S)_{\{j\}}$ for all $j$.
\end{example}

Consistency can be understood as defined with respect to a relaxation of $\S$.  For classical consistency, the relaxation is $\S_J$, obtained by omitting constraints from $\S$.  We will later define consistency with respect to the linear programming relaxation of $\S$.

\section{Consistency and Backtracking}

It is well known that consistency is closely related to backtracking.  We note first that branching can find a feasible solution for a fully consistent constraint set without backtracking, assuming of course that the constraints have a solution.  Suppose we branch on variables $x_1,\ldots, x_n$ in that order.  Each node in level $j$ of the branching tree corresponds to a partial assignment $(x_1,\ldots,x_{j-1})=(v_1,\ldots,v_{j-1})$.  We branch on $x_j$ at the node by assigning to $x_j$ each value $v_j\in D_j$ for which the partial assignment $(x_1,\ldots,x_j)=(v_1,\ldots,v_j)$ violates no constraint in $\S$.  Due to the consistency of $\S$, this partial assignment is consistent with $\S$ for at least one value $v_j\in D_j$.  Thus branching can continue to the bottom of the tree with no need to backtrack.

%It is clear from this argument that a weaker form of consistency avoids backtracking.  It is {\em local consistency}, which requires that any consistent partial assignment extend to one more variable.  Thus $\S$ is locally consistent if for every $J\subseteq N$ and every $j\in N\setminus J$, $D_J(\S) = D_{J\cup\{j\}}(\S)|_J$.  

A weaker form of consistency avoids backtracking if there is limited coupling of variables.  Let the directed {\em dependency graph} $G$ of $\S$ for the ordering $1,\ldots, n$ consist of vertices corresponding to variables $x_j$ and directed edges $(x_i,x_j)$ whenever $i<j$,  and $x_i$ and $x_j$ occur in a common constraint of $\S$.  The {\em width} of $G$ for the ordering $1,\ldots,n$ is the maximum out-degree of the nodes of $G$. 
\begin{definition}
A constraint set $\S$ is \mbox{\em $k$-consistent} if $D_J(\S_J)=D_{J\cup\{j\}}(\S_{J\cup{\{j]}})|_J$ for all \mbox{$J\subseteq N$} with $|J|=k-1$ and all $j\in N\setminus J$.  $\S$ is {\em strongly \mbox{$k$-consistent}} if it is $j$-consistent for $j=1,\ldots,k$.  
\end{definition}
The following is proved in \cite{Fre82}.
\begin{proposition} \label{prop:freuder}
Let $G$ be the directed dependency graph of constraint set $\S$ for the ordering $1,\ldots, n$.  Then if the branching order is $x_1, \ldots, x_n$, $\S$ can be solved without backtracking if $\S$ is strongly \mbox{$k$-consistent} and $G$ has width less than $k$.
\end{proposition}

A still weaker form of consistency avoids backtracking if the branching order is given.  It is not necessary to consider all sets $J$ and all indices $j\not\in J$, but only variables on which we have branched. We therefore define a form of $k$-consistency that assumes the branching order is $x_1, \ldots, x_n$.   Let $J_k=\{1, \ldots, k\}$.
\begin{definition}
A 0--1 constraint set $\S$ is {\em sequentially $k$-consistent if $D_{J_{k-1}}(\S_{J_{k-1}})=D_{J_k}(\S_{J_{k}})|_{J_{k-1}}$.}  
\end{definition}
Thus $\S$ is sequentially $k$-consistent if for every partial assignment $(x_1,\ldots,x_{k-1})=(v_1,\ldots,v_{k-1})$ that violates no constraint in $S$, there is a value $v_k$ in $D_k$ such that $(x_1,\ldots,x_k)=(v_1,\ldots,v_k)$ violates no constraint in $\S$.  The following is easy to show.  
\begin{proposition} \label{prop:orderconsistent}
If the branching order is $x_1,\ldots, x_n$, constraint set $\S$ can be solved without backtracking if $\S$ is sequentially $k$-consistent for $k=1,\ldots,n$.
\end{proposition}

\begin{example}\label{ex:orderconsistency}
Let $\S=\{3x_1+2x_2\geq 1,\; -x_1+2x_2\geq 0, \; x\in\{0,1\}^2\}$.  Proposition~\ref{prop:orderconsistent} implies that we can avoid backtracking by branching in the order $x_1,x_2$, because $\S$ is sequentially 1-consistent and sequentially 2-consistent.  The lack of backtracking does not follow from Proposition~\ref{prop:freuder}, however, because $\S$ is not \mbox{2-consistent}, and its dependency graph has width 1 for the ordering 1,2.  $\S$ is not 2-consistent because the partial assignment $x_2=0$ violates no constraints and has no extension to a consistent assignment $(x_1,x_2)=(v_1,0)$.  
\end{example}

Even domain consistency suffices to avoid backtracking if it is achieved at every node.  Supposing that a given node corresponds to a partial assignment as above, let $\S'=\S\cup \{(x_1,\ldots, x_{j-1})=(v_1,\ldots,v_{j-1})\}$.  Then if $\S'$ is domain consistent, $x_j$ can be assigned any value in its domain to obtain an assignment $(x_1,\ldots,x_j)=(v_1,\ldots, v_j)$ that is consistent with $\S'$.  The process can continue without backtracking if domain consistency is similarly achieved at subsequent nodes.

\section{Consistency and Resolution}

Previous research has shown that the resolution procedure of propositional logic achieves consistency for a 0--1 constraint set.  

First, some definitions.  A {\em literal} $\ell_j$ is a proposition of the form $x_j$ or $\neg x_j$.  A logical {\em clause} is a disjunction $\bigvee_{j\in J} \ell_j$ of literals, which we denote by $\ell(J)$, where $J$ is possibly empty.   Given clauses $\ell(J)$ and $\ell(J')$, the former {\em absorbs} the latter if $J\subseteq J'$.  For example, $x_1\vee \neg x_2$ absorbs $x_1\vee \neg x_2\vee x_3$.  One clause logically implies another if and only if the one absorbs the other.  

Given clauses $\ell(J_1)\vee x_k$ and $\ell(J_2)\vee \neg x_k$, where $k\not\in J_1\cup J_2$, the {\em resolvent} of the clauses is $\ell(J_1\cup J_2)$.  For example, the resolvent of $x_1\vee x_2\vee x_4$ and $x_1\vee \neg x_3\vee \neg x_4$ is $x_1\vee x_2\vee\neg x_3$, while the clauses $x_1\vee \neg x_2$ and $\neg x_1\vee x_2$ do not have a resolvent.  A resolvent is logically implied by the conjunction of its two parents but is absorbed by neither.  

Given a clause set $\C$, a {\em resolution proof} of clause $C_m$ from $\C$ is a sequence of clauses $C_1,\ldots, C_m$ such that $C_m$ is the resolvent of two previous clauses in the sequence, and each $C_i$ for $i=1,\ldots, m-1$ either belongs to $\C$ or is the resolvent of two earlier clauses in the sequence. It can be assumed that a resolvent is not generated when it is absorbed by a previous clause in the sequence. An {\em input proof} of $C_m$ is a resolution proof in which at least one of the parents of each resolvent belongs to $\C$ \cite{Cha70}.  There is a resolution proof of any clause that is logically implied by $\mathcal{C}$ \cite{Qui52,Qui55}, but not necessarily an input proof.

A 0--1 constraint set $\S$ {\em logically implies} 0--1 constraint set $\S'$ when all 0--1 points that satisfy $\S$ also satisfy $\S'$.  $\S$ and $\S'$ are {\em logically equivalent} when they logically imply each other.  A logical clause 
\[
\bigvee_{j\in J^+} \hspace{-1ex} x_j \vee \bigvee_{j\in J^-} \hspace{-1ex} \neg x_j
\]
is {\em represented} by the 0--1 inequality 
\[
\sum_{j\in J^+} \hspace{-0.5ex} x_j + \sum_{j\in J^-} \hspace{-0.1ex} (1-x_j)\geq 1
\]
A 0--1 inequality is {\em clausal} when it represents a clause.  It is clear that a 0--1 inequality is logically equivalent to the set of clausal inequalities it implies.  Thus if we let $\S_{\Cr}$ be the set of clausal inequalities that are implied by some inequality in $\S$, then $\S$ is logically equivalent to $\S_{\Cr}$.  It is convenient to say that a clausal inequality has a resolution proof from $\S_{\Cr}$ if the clause it represents has a resolution proof from the clauses represented by $\S_{\Cr}$. It is shown in \cite{Hooker12} that resolution on clausal inequalities achieves consistency.  
\begin{proposition} \label{prop:resolution}
If 0--1 constraint set $\S$ is augmented with all clausal inequalities that have resolution proofs from $\S_{\Cr}$, the resulting constraint set is consistent.  
\end{proposition}

\begin{example}
If $\S$ is the constraint set of Example~\ref{ex:orderconsistency}, $\S_{\Cr}$ contains the clausal inequalities $x_1+x_2\geq 1$ and $-x_1+x_2\geq 0$.  One clausal inequality, namely $x_2\geq 1$, has a resolution proof from $\S_{\Cr}$, and adding this inequality to $\S$ yields a consistent constraint set.  
\end{example}

A second example illustrates how a traditional cutting plane can serve the dual purpose of tightening the linear programming (LP) relaxation and achieving consistency.  Let the LP relaxation of $\S=\{Ax\geq b, \; x\in \{0,1\}^n\}$ be $\S_{\LP}=\{Ax\geq b,\; x\in [0,1]^n\}$.

\begin{example} \label{ex:3constraints}
Suppose that $\S$ is the constraint set of Example~\ref{ex:0}.  
In this case, $\S$ and $\S_\Cr$ are identical.  Resolution yields two additional clausal inequalities, $x_1+x_2\geq 1$ and $x_1+x_3\geq 1$.  By Proposition~\ref{prop:resolution}, adding these inequalities to $\S$ achieves consistency.  These inequalities are also traditional cutting planes for $\S$, in particular Chv\'{a}tal-Gomory (C-G) cuts.  The first cuts off two fractional vertices $(x_1,\ldots,x_4)=(\frac{1}{3},\frac{1}{3},0,\frac{1}{3}), (\frac{1}{2},0,0,\frac{1}{2})$ of the polytope described by $\S_{\LP}$, and the second cuts off $(\frac{1}{2},\frac{1}{2},0,0)$ as well.  The inequalities therefore serve the dual purpose of achieving consistency and tightening the LP relaxation.  As it happens, adding both resolvents yields an integral polytope, but we will see that a consistent constraint does not in general describe an integral polytope. 
\end{example}

Input proofs from $\S_{\Cr}$ do not necessarily achieve consistency, but they derive all clausal C--G cuts for $\S_{\Cr}$.  
%Given a set $\mathcal{T}$ of linear inequalities, let $\mathcal{T}^*$ denote the union of $\mathcal{T}$ with the set of clausal inequalities that are C--G cuts for $\mathcal{T}\cup\{x\in[0,1]\}$.  
The following is proved in \cite{Hoo89}.
\begin{proposition}
%Given a 0--1 constraint set $\S$, $\S_{\Cr}^*$ is precisely the set of clausal inequalities that have input proofs from $\S_{\Cr}$.
Given a 0--1 constraint set $\S$, a clausal inequality is a C-G cut for $\S_{\Cr}\cup\{x\in [0,1]\}$ if and only if it has an input proof from $\S_{\Cr}$.
\end{proposition}

\section{LP-consistency}

While resolution can always achieve consistency, it is not a practical method for the reduction of backtracking.  Resolution proofs tend to explode rapidly in length and complexity.  However, the LP relaxation of $\S$ provides an additional tool for this purpose.  Specifically, it provides a more useful test for consistency than whether a partial assignment violates a constraint.

Consistency of $\S$ implies that any partial assignment $x_J=v_J$ that is consistent with $\S_J$ (i.e., violates no constraint in $\S$) is consistent with $\S$.  We want a type of consistency that ensures that any partial assignment consistent with $\S_{\LP}$ is consistent with $\S$.  We can achieve this by defining consistency with respect to the LP relaxation $\S_{\LP}$ rather than the relaxation $\S_J$.  Recall that classical consistency is defined so that $D_J(\S_J)=D(S)|_J$.  We therefore define {\em LP-consistency} as follows.
\begin{definition}
%A 0--1 constraint set $\S$ is {\em LP-consistent} if $D_J(\S_{\LP}) \cap \{0,1\}^{|J|}=D(\S)|_J$ for all $J\subseteq N$.
A 0--1 constraint set $\S$ is {\em LP-consistent} if $D_J(\S_{\LP})=D(\S)|_J$ for all $J\subseteq N$.  
\end{definition}
Here, $D_J(\S_{\LP})$ refers to the set of {\em 0--1 assignments} to $x_J$ that are consistent with $\S_{\LP}$.  Thus $\S$ is {\em LP-consistent} if $\S_{\LP} \cup \{x_J = v_J\}$ is infeasible for any 0--1 partial assignment $x_J = v_J$ that is inconsistent with $\S$.

\begin{example} \label{ex:LPconsistency}
Consider the 0--1 constraint set $\S = \{2x_1 + 4x_2 \geq -1, \; 2x_1-4x_2\geq -3, \; x\in\{0,1\}^2\}$ (Fig.~\ref{fig:LPconsistency}).  The partial assignment $x_1=0$ is consistent with $\S_{\LP}$ but not with $\S$, because both $(x_1,x_2)=(0,0)$ and $(x_1,x_2)=(0,1)$ violate $\S$.  So $\S$ is not LP-consistent. 
\end{example}

\begin{figure}[!t]
\centering
\includegraphics[scale=0.6,clip,trim=240 140 290 150]{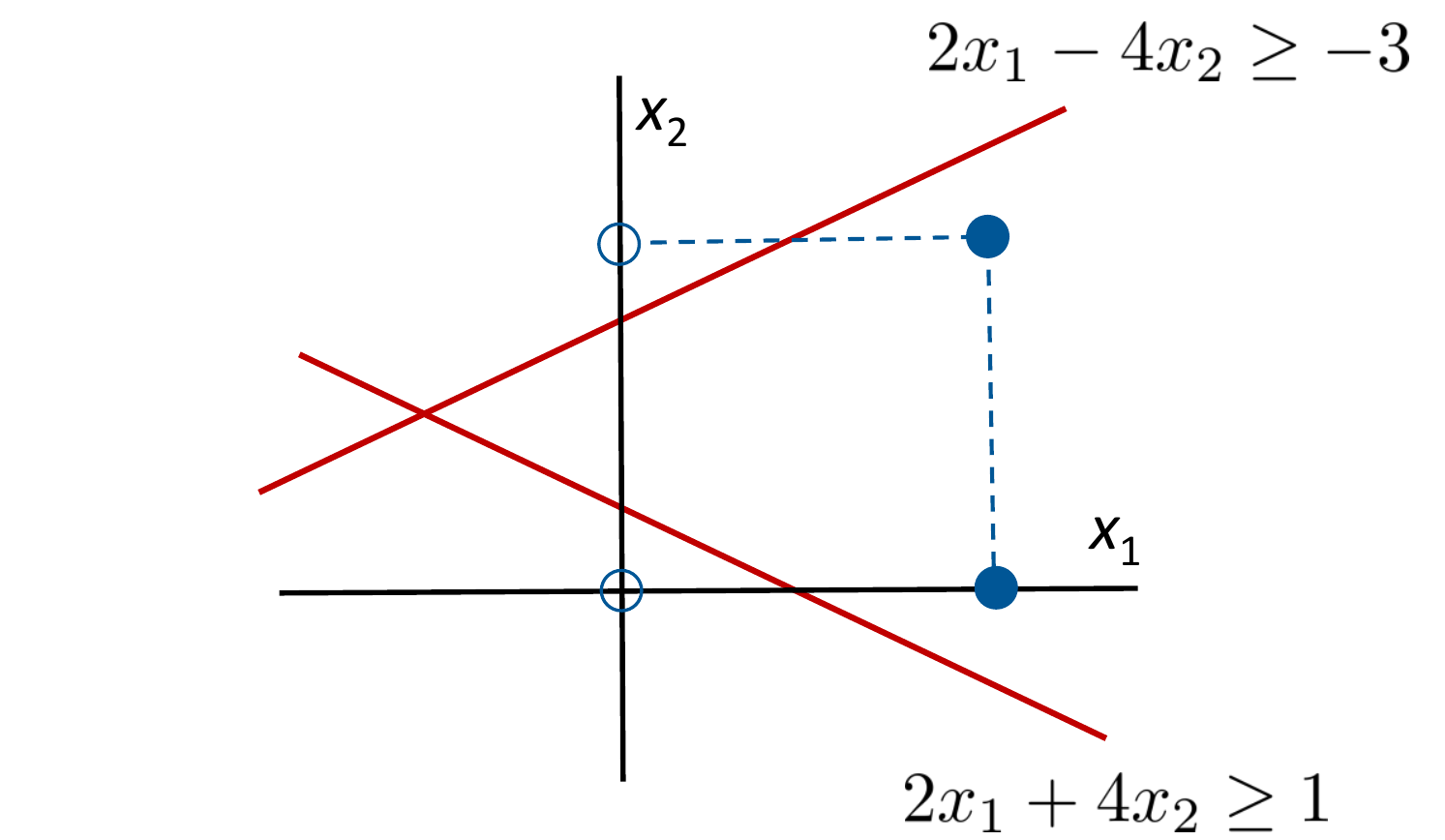}
\caption{Illustration of Example~\ref{ex:LPconsistency}.}\label{fig:LPconsistency}
\end{figure}

Two elementary properties of LP-consistency follow.
\begin{proposition}\label{prop:LPconsistent}
A consistent 0--1 constraint set is LP-consistent.
\end{proposition}
\begin{proof}
	Consider any 0--1 partial assignment $x_J=v_J$ that is consistent with $S_{\LP}$.  We claim that $x_J=v_J$ is consistent with $\S$, which suffices to show that $\S$ is LP-consistent.  Since $S_{\LP}\cup\{x_J=v_J\}$ is feasible, $x_J=v_J$ violates no constraints in $\S$.  Now since $\S$ is consistent, this means that $x_J=v_J$ is consistent with $\S$, as claimed. \qed
\end{proof}
%\begin{proof} If $\S$ is consistent, we have
%\begin{equation}
%L_J(S)\subseteq D_J(S)=D(S)|_J \subseteq L_J(S)  \label{eq:LPconsistent}
%\end{equation}
%which implies $L_J(\S)=D(\S_J)$.  The last inclusion in (\ref{eq:LPconsistent})
%is due to the fact that $D(S)|_J$ is the set of integer partial assignments $x_J=v_J$ such that $S\cup\{x_J=v_J\}$ is feasible, and $L_J(S)$ is the set of integer partial assignments such that $S_{\mathrm{LP}}\cup \{x_J=v_J\}$ is feasible, and the latter constraint set is a relaxation of the former. $\Box$
%\end{proof}
In addition, a 0--1 constraint set that describes the integer hull (the convex hull of feasible 0--1 points) is LP-consistent.
\begin{proposition}\label{prop:convexHull}
Given 0--1 constraint set $\S$, if $\S_{\LP}$ describes the integer hull of $D(\S)$, then $\S$ is LP-consistent.
\end{proposition}
\begin{proof}
    Suppose that $\S\cup\{x_J=v_J\}$ is infeasible for a given 0--1 partial assignment $x_J=v_J$.  Then $x_J=v_J$ describes a face of the unit hypercube that is disjoint from $D(\S)$.  This implies that the face is disjoint from the convex hull of $D(\S)$, which is described by $\S_{\LP}$.  Thus $\S_{\LP}\cup\{x_J=v_J\}$ is infeasible, and it follows that $\S$ is LP-consistent. \qed
\end{proof}

It is essential to observe that a convex hull model is not necessary to achieve LP-consistency, a fact that will be exploited in later sections.  This can be seen in an example.

\begin{example} \label{ex:convexHull}
Consider the following two constraint sets (Fig.~\ref{fig:convexHull}), which have the same feasible set:
\[
\begin{array}{l}
\S^1 = \{x_1+x_2\leq 1, \; x_2+x_3\leq 1, \; x\in \{0,1\}^3\} \vspace{0.5ex} \\
\S^2 = \{x_1+2x_2+x_3\leq 2, \; x\in\{0,1\}^3\}
\end{array}
\]
The LP relaxation $\S^1_{\LP}$ describes the integer hull of $D(\S^1)=D(\S^2)$, and so $\S^1$ is LP-consistent by Proposition~\ref{prop:convexHull}.  Yet the constraint set $\S^2$ is also LP-consistent, even though $\S^2_{\LP}$ does not describe the convex hull, but describes a polytope with fractional extreme points $(x_1,x_2,x_3)=(0,\frac{1}{2},1),(1,\frac{1}{2},0)$.  Interestingly, the inequality $x_1+2x_2+x_3\geq 2$ in $\S^2$ is the sum of the two nontrivial facet-defining inequalities in $\S^1$ and is therefore weaker than either of them from a polyhedral point of view.  Yet it cuts off more infeasible 0--1 points than either of the facet-defining inequalities and is therefore stronger in this sense.  Indeed, the purpose of achieving LP-consistency is to cut off infeasible 0--1 (partial) assignments, not to cut off fractional vertices of the LP relaxation. 
\end{example}

\begin{figure}
\centering
\includegraphics[scale=0.6,clip,trim=240 80 290 150]{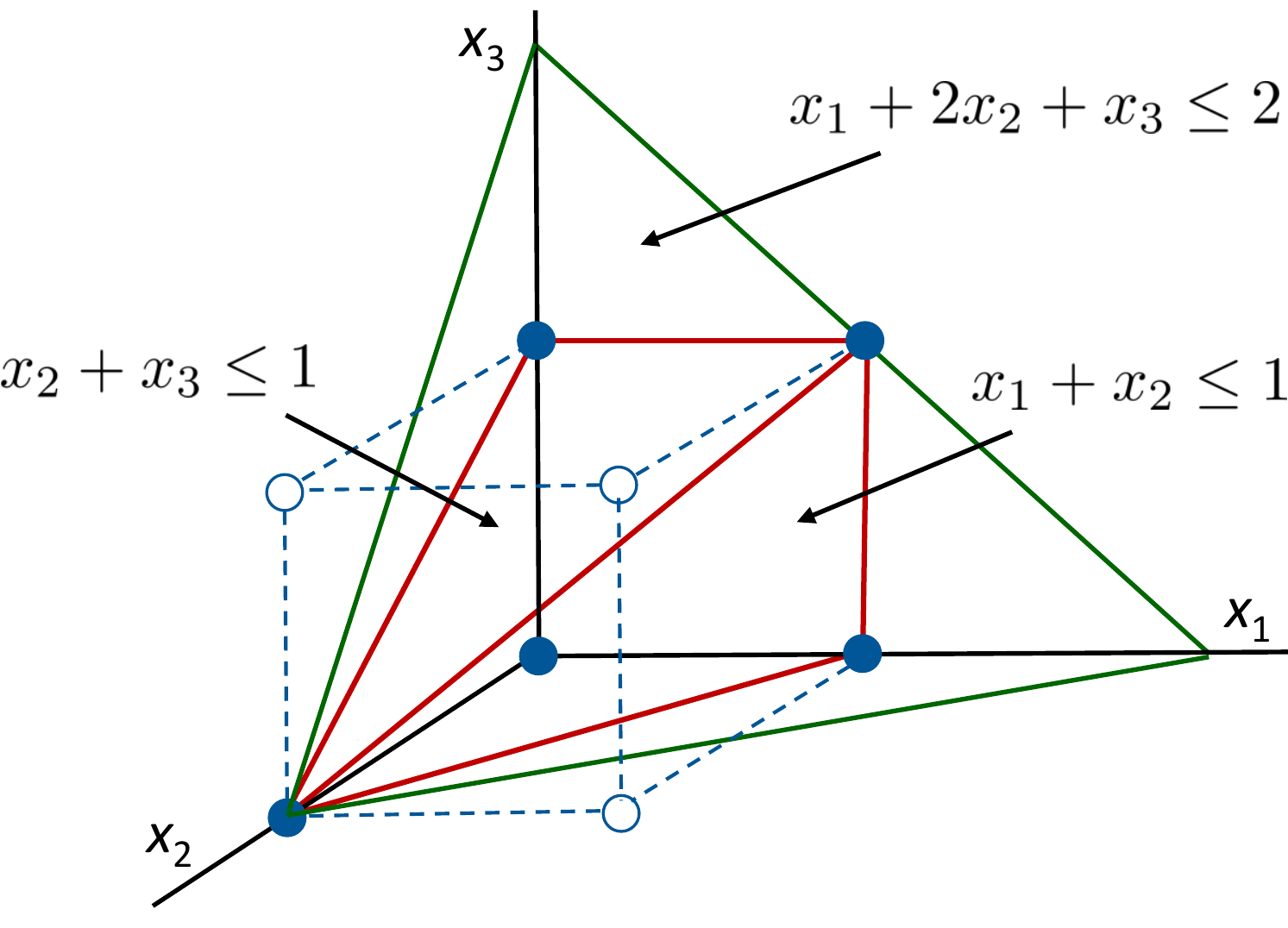}
\caption{Illustration of Example~\ref{ex:convexHull}}\label{fig:convexHull}
\end{figure}

\section{Characterizing LP-Consistency} \label{sec:characterizing}

%The advantage of a consistent constraint set $\S$ is that it provides an easily computable test for whether a partial assignment is consistent with $\S$.  This, in turn, makes it easier to avoid backtracking.  One measure of the power of a consistency property is {\em how} easy it is to apply the test.  When $\S$ is fully consistent, the test is whether $x_J=v_J$ violates a constraint in $\S$, which is trivial to check.  When $\S$ has no consistency properties, the test is simply whether $x_J=v_J$ is consistent with $\S$, which is hard to check.  When $\S$ is LP-consistent, the test is whether $x_j=v_J$ is consistent with $\S_{\LP}$, which is nontrivial to check but still rather easy because it is a matter of solving an LP.  

%Another measure of power is how few clausal inequalities must be checked to show that a partial assignment passes the test.  When $\S$ is fully consistent, the test requires that $x_J=v_J$ violate no clausal inequality logically implied by an individual constraint in $\S$.  When $\S$ has no consistency properties, the test requires that $x_J=v_J$ violate no clausal inequality implied by $\S$ as a whole.  When $\S$ is LP-consistent, the test requires that $x_J=v_J$ violate no clausal C-G cut for $\S_{\LP}$.  This last fact is due to the following.  

The following result gives a necessary condition for consistency based on clausal inequalities.
\begin{proposition}
If a constraint set $\S$ is consistent, then all of its implied clausal inequalities are in $\S_{\Cr}$.  
\end{proposition}

\begin{proof}
Suppose that $\S$ is consistent, and let $C$ be any clausal inequality implied by $\S$.  Then the assignment $x_J=v_J$ violates $C$, where $x_J$ are the variables in $C$ and $v_j$ is 1 when $x_j$ is negated in $C$ and 0 otherwise. This means $x_J=v_J$ is inconsistent with $\S$, which implies by the consistency of $\S$ that $x_J=v_J$ violates an inequality $\alpha x \geq \beta$ in $\S$. As a result, $C$ must be implied by $\alpha x \geq \beta$, showing that $C \in \S_{\Cr}$.  \qed
\end{proof}

LP-consistency allows us to derive a stronger argument on the relation between an LP-consistent set and its implied clausal inequalities, as it provides both necessary and sufficient conditions. In particular, a 0--1 constraint set $\S$ is LP-consistent if and only if all of its implied clauses are C-G cuts for $\S_{\LP}$.  This is due to the following fact.

\begin{proposition} \label{prop:CC}
	Given a 0--1 constraint set $\S$, a 0--1 partial assignment is consistent with $\S_{\LP}$ if and only if the assignment violates no clausal C-G cut for $\S_{\LP}$.
\end{proposition}

\begin{proof}
	It suffices to show that a given 0--1 partial assignment $x_J=v_J$ violates a clausal C-G for $\S_{\LP}$ if and only if $\S_{\LP}\cup\{x_J=v_J\}$ is infeasible.  Suppose first that $x_J=v_J$ violates a clausal inequality $ax\geq\beta$ that is a C-G cut for $\S_{\LP}$, where $S_{\LP}$ is the system $Ax\geq b$.  Since $x_J=v_J$ violates $ax\geq\beta$, we can write the inequality as $a_Jx_J\geq\beta$, where \mbox{$a_Jv_J\leq \beta-1$}.  Now since $ax\geq\beta$ is a C--G cut, there is a tuple $u\geq 0$ of multipliers such that $uA=a$ and $\beta-1< ub \leq \beta$.  We therefore have $(uA)_Jv_J=a_Jv_J \leq \beta-1<ub$.  This implies that $x_J=v_J$ violates $uAx\geq ub$, and so $\S_{\LP}\cup\{x_J=v_J\}$ must be infeasible.  
	
	For the converse, suppose that $\S_{\LP}\cup\{x_J=v_J\}$ is infeasible, which means that the face of the unit hypercube defined by $x_J=v_J$ lies outside the polytope defined by $S_{\LP}$.  Let $J^+=\{j\in J\;|\; v_j=0\}$ and $J^-=\{j\in J\;|\; v_j=1\}$.  
	Then some inequality of the form $\sum_{j\in J^+} x_j + \sum_{j\in J^-} (1-x_j) \geq \bar{\pi}$ for some $\bar{\pi} > 0$ separates the face just mentioned from the polytope; i.e., $x_J=v_J$ violates this inequality. Since this inequality is valid for $\S_{\LP}$, it is dominated by some surrogate of $Ax\geq b$. That is there exists a tuple $u \geq 0$ of multipliers such that $uA \geq ub$ is of the form
	\begin{equation}
	\sum_{j\in J^+} x_j + \sum_{j\in J^-} (1-x_j) \geq \pi
	\label{eq:separator}
	\end{equation}
	where $\pi \geq \bar{\pi}$, and $\pi\leq |J|$ because $\S$ is feasible. Now pick any subset $\hat{J}\subseteq J$ with $|\hat{J}|=\lceil \pi \rceil - 1$, let $\hat{J}^+=J^+\cap\hat{J}$, and let $\hat{J}^-=J^-\cap\hat{J}$.  Take the sum of (\ref{eq:separator}) with $-x_j\geq -1$ for $j\in \hat{J}^+$ and  $x_j\geq 0$ for $j\in \hat{J}^-$.  This yields a clausal inequality that is a surrogate of $Ax\geq b$:
	\[
	\sum_{j\in J^+\setminus\hat{J}^+} \hspace{-2ex} x_j + \hspace{-1ex} \sum_{j\in J^-\setminus\hat{J}^-} \hspace{-2ex} (1-x_j) \geq 1 + \pi - \lceil \pi \rceil
	\]
	Rounding up the right-hand side (if necessary) yields a clausal C--G cut violated by $x_J=v_J$.  Thus $x_J=v_J$ violates a clausal C-G cut for $\S_{\LP}$, as claimed.
	\qed
\end{proof}

\begin{example}
Consider again the constraint set $\S$ of Example~\ref{ex:3constraints}.  The partial assignment $(x_1,x_3)=(0,0)$ is inconsistent with $\S_{\LP}$ and violates a clausal C-G cut, namely $x_1+x_3\geq 1$.  The cut is obtained by assigning multipliers $\frac{1}{4}, \frac{1}{2}, \frac{1}{4}, \frac{1}{4}, \frac{1}{2}$ to the three constraints of $\S$, $x_2\geq 0$, and $x_3\geq 0$, respectively.  The partial assignment $(x_1,x_3)=(0,1)$ is consistent with $\S_{\LP}$ and therefore violates no clausal C-G cut.
\end{example}

%We can also measure the power of a consistency property by comparing it with the power of an inference method.  Let us say that a 0--1 inequality {\em dominates} a clause when it implies the clause.  Thus $\S_{\Cr}$ consists of the clauses dominated by individual inequalities in $\S$.  In view of the above discussion, we can say that full consistency has the same power as domination plus resolution, in the following sense: the test for whether $x_J=v_J$ is consistent with $\S$ is whether it violates a clause that has a resolution proof from $\S_{\Cr}$.  The absence of a consistency property has the same power as domination alone, because the test is whether $x_J=v_J$ violates a clause in $\S_{\Cr}$.  LP-consistency has the same power as C-G cutting planes, because the test is whether $x_J=v_j$ violates a clausal C-G cut for $\S_{\LP}$.  

\begin{corollary} \label{cor:LPconsistency} A constraint set $\S$ is LP-consistent if and only if all of its implied clausal inequalities are C-G cuts for $\S_{\LP}$. 
\end{corollary}

\begin{proof}
Suppose first that $\S$ is LP-consistent, and let $C$ be any clausal inequality implied by $\S$.  Then the assignment $x_J=v_J$ violates $C$, where $x_J$ are the variables in $C$ and $v_j$ is 1 when $x_j$ is negated in $C$ and 0 otherwise.  This means $x_J=v_J$ is inconsistent with $\S$, which implies by the LP-consistency of $\S$ that $x_J=v_J$ is inconsistent with $\S_{\LP}$. By Proposition~\ref{prop:CC}, $x_J=v_J$ violates some clausal C-G cut $C'$ of $\S_{\LP}$.  Then $C'$ must absorb $C$, which means $C$ is likewise a C-G cut of $\S_{\LP}$.

Conversely, suppose all clausal inequalities implied by $\S$ are C-G cuts for $\S_{\LP}$, and consider any partial assignment $x_J=v_J$ that is consistent with $\S_{\LP}$.  By Proposition~\ref{prop:CC}, $x_J=v_J$ violates no clausal C-G cut of $\S_{\LP}$.  This means that it violates no clause implied by $\S$, which implies that $x_J=v_J$ is consistent with $\S$, as desired. \qed
\end{proof}

\begin{example}
The constraint set $\S$ of Example~\ref{ex:0} is LP-consistent because its implied clausal inequalities are all absorbed by the inequalities in $\S\cup\{x_1+x_2\geq 1, \; x_1+x_3\geq 1\}$, and these are all C-G cuts for $\S_{\LP}$.
\end{example}

\section{LP-consistency and Resolution}

We have seen that full resolution achieves consistency for a 0--1 constraint set $\S$.  That is, $\S$ is consistent if it contains all clausal inequalities that have resolution proofs from $\S_{\Cr}$.  Full resolution therefore achieves LP-consistency, since any consistent constraint set is LP-consistent (Proposition~\ref{prop:LPconsistent}).  However, it is generally unnecessary to apply full resolution to achieve LP-consistency, and it is unclear what kind of resolution is necessary and sufficient for this purpose.     

Yet one can go some distance toward achieving LP-consistency by generating all clausal inequalities that have input proofs from $\S_{\Cr}$.  
%In fact, LP-consistency cannot be achieved without generating these clausal inequalities, in the following sense. 

%\begin{proposition}
%Given a 0--1 constraint set $\S$, adding clausal inequalities to $\S_{\Cr}$ to obtain a larger set $\S'_{\Cr}$ achieves LP-consistency for $\S'_{\Cr}$ only if $\S_{\Cr}$ contains all clausal inequalities with input proofs from $\S_{\Cr}$, or clauses that absorb them.
%\end{proposition}

%\begin{proof}
%Suppose that $\S'_{\Cr}$ contains no clause that absorbs a clause $C$ with an input proof from $\S_{\Cr}$.  It suffices to show that $\S'_{\Cr}$ is not LP-consistent.  Let $x_j=v_J$ be the 0--1 partial assignment that violates $C$, where $x_C$ contains the variables of $C$.  Since $x_J=v_J$ violates $C$, it is inconsistent with $S_{\Cr}$ and therefore with $\S'_{\Cr}$.  Since $\S'_{\Cr}$ contains no clause that absorbs $C$, by Proposition xx, $\S'_{\Cr}$ contains no C-G clausal cut that $x_J=v_J$ violates.  This means by proposition xx that $x_J=v_J$ is consistent with $S\_{\Cr}$.  Thus we have a partial assignment that is consistent with $\S'_{\Cr}$ but is inconsistent with $\S'_{\Cr}\cup \{x\in \{0,1\}^n\}$.  This shows that $\S'_{\Cr}$ is not LP-consistent, as desired.
%\end{proof}

\begin{example}
Let $\S$ consist of the following constraints, in addition to $x\in \{0,1\}$:
\[
\begin{array}{l@{\hspace{10ex}}l}
x_1 + x_2 + x_3 + x_4 \geq 1 & x_1 - x_2 + x_3 + x_4 \geq 0 \\
x_1 + x_2 + x_3 - x_4 \geq 0 & x_1 - x_2 + x_3 - x_4 \geq -1 \\
x_1 + x_2 - x_3 + x_4 \geq 0 & x_1 - x_2 - x_3 + x_4 \geq -1 \\
x_1 + x_2 - x_3 - x_4 \geq -1 & x_1 - x_2 - x_3 - x_4 \geq -2 
\end{array}
\]
In this case, $\S_{\Cr}=\S$.  Input proofs from $\S_{\Cr}$ yield the inequalities
\[
\begin{array}{l@{\hspace{6ex}}l@{\hspace{6ex}}l}
x_1 + x_2 + x_3 \geq 1 & x_1 + x_2 + x_4 \geq 1 & x_1 + x_3 + x_4 \geq 1 \\
x_1 + x_2 - x_3 \geq 0 & x_1 + x_2 - x_4 \geq 0 & x_1 + x_3 - x_4 \geq 0 \\
x_1 - x_2 + x_3 \geq 0 & x_1 - x_2 + x_4 \geq 0 & x_1 - x_3 + x_4 \geq 0 \\ 
x_1 - x_2 - x_3 \geq -1 & x_1 - x_2 - x_4 \geq -1 & x_1 - x_3 - x_4 \geq -1 
\end{array}
\]
Adding these inequalities to $\S_{\Cr}$ to obtain $\S'$ does not achieve LP-consistency.  The partial assignment $x_1=0$ is inconsistent with $\S$ but consistent with $\S'_{\LP}$, where the latter is due to the fact that $(x_1,\ldots,x_4)=(0,\frac{1}{2},\frac{1}{2},\frac{1}{2})$ satisfies $\S'_{\LP}$.  Yet the input proofs are not useless because they eliminate some partial assignments that are inconsistent with $\S$, such as $(x_1,x_2)=(0,0)$, which is consistent with $\S_{\Cr}$ but not with $\S'_{\LP}$.  
%A partial assignment $x_J=v_J$ of this kind can be avoided in branching because the LP constraint set $\S'_{\Cr}\cup\{x_J=v_J\}$ is infeasible. 
\end{example}

On the other hand, it may be possible to achieve LP-consistency without adding all clauses that have input proofs from $\S_{\Cr}$, and even without adding all the clauses in $\S_{\Cr}$.

\begin{example} \label{ex:counterexample}
Let $\S$ be the constraint set of Example~\ref{ex:LPconsistency}, where $\S_{\Cr}=\{x_1+x_2\geq 1, \; x_1-x_2\geq 0\}$.  The clausal inequality $x_1\geq 1$ has an input proof from $\S_{\Cr}$, but $\S_{\Cr}$ is LP-consistent without adding this inequality.  In fact, we can achieve LP-consistency for $\S$ even without adding to it all the clausal inequalities in $\S_{\Cr}$.  For example, if we add only $x_1+x_2\geq 1$ to $\S$, the resulting constraint set is LP-consistent (Fig.~\ref{fig:LPconsistency2}).
\end{example}

\begin{figure}[!b]
\centering
\includegraphics[scale=0.6,clip,trim=240 140 290 150]{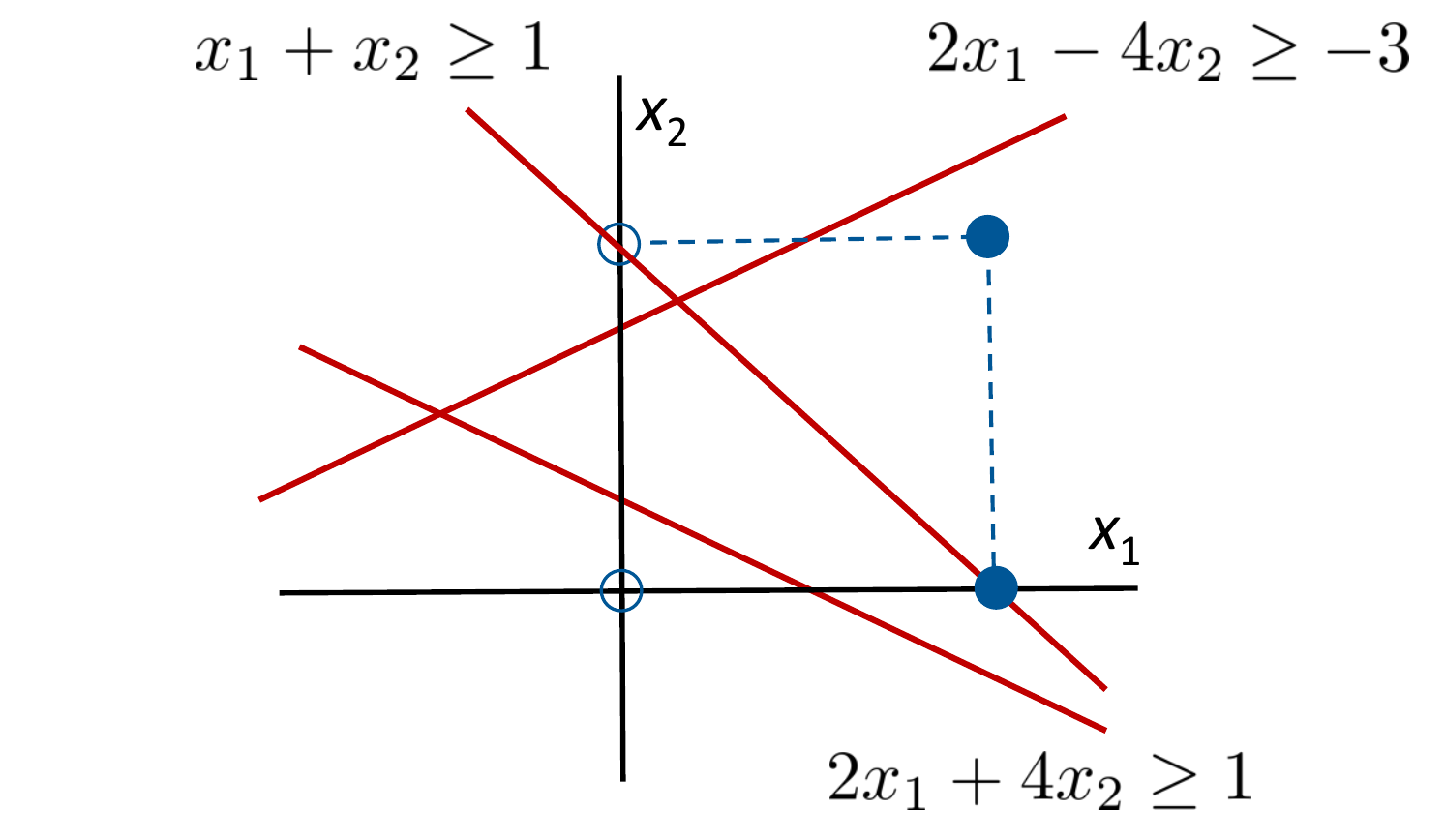}
\caption{Illustration of Example~\ref{ex:counterexample}.}\label{fig:LPconsistency2}
\end{figure}

These examples suggest that resolution is not a natural technique for achieving LP-consistency or even partial LP-consistency.  However, we will see in Section~\ref{LPconsistency} that a technique borrowed from integer programming can achieve partial LP-consistency in a reasonably practical way.

%Recall from Propositon xx that any clausal inequality that has an input proof from $\S_{\Cr}$ is a C-G cut for $S_{\Cr}\cup \{x\in [0,1]\}$.  But we also know from Proposition xx that a partial assignment is consistent with $\S_{\LP}$ if and only if it violates no clausal C-G cut for $\S_{\LP}$

\section{LP-Consistency and Backtracking} \label{sec:partial consistency}

Like full consistency in CP, full LP-consistency is difficult to achieve.  We therefore follow the lead of the CP community and consider weaker forms of consistency.  One that seems appropriate to 0--1 programming is inspired by $k$-consistency.  While even $k$-consistency is hard to achieve in practice, and the CP community focuses on domain consistency instead, a form of LP-consistency analogous to sequential $k$-consistency may be practical for 0--1 programming.

Recall that $\S$ is sequentially $k$-consistent if $D_{J_{k-1}}(\S_{J_{k-1}})=D_{J_k}(\S_{J_k})|_{J_{k-1}}$, and that sequential $k$-consistency for $k=1,\ldots, n$ suffices to avoid backtracking when the branching order is $x_1, \ldots, x_n$.  A parallel definition that relates to linear programming is as follows.
\begin{definition} \label{def:LPconsistency}
A 0--1 constraint set $\S$ is {\em sequentially LP $k$-consistent} if $D_{J_{k-1}}(\S_{\LP})=D_{J_k}(\S_{\LP})|_{J_{k-1}}$.  
\end{definition}
Equivalently, we can say that $\S$ is sequentially LP $k$-consistent if for every 0--1 partial assignment $x_{J_{k-1}}=v_{J_{k-1}}$ that is consistent with $\S_{\LP}$, there is a 0--1 assignment $x_k=v_k$ for which $x_{J_k}=x_{J_k}$ is consistent with $\S_{\LP}$.  Thus sequential LP $k$-consistency is analogous to sequential $k$-consistency but based on the $\S_{\LP}$ relaxation rather than the $S_{J_{k-1}}$ relaxation.

This form of consistency can also allow us to avoid backtracking, if we are willing to solve appropriate LP problems. Specifically, suppose that at a given node in the branching tree, prior branching has fixed $(x_1,\ldots,x_{k-1})=(v_1,\ldots,v_{k-1})$.  For the next branch, we select a value \mbox{$v_k\in\{0,1\}$} for which the partial assignment $(x_1,\ldots,x_k)=(v_1,\ldots,v_k)$ is consistent with $\S_{\LP}$; that is, for which the LP problem $S_{\LP}\cup\{(x_1,\ldots,x_k)=(v_1,\ldots,v_k)\}$ is feasible.  We then set $x_k=v_k$ and continue to the next level of the tree.  The following theorem guarantees that the LP problem will be feasible for at least one value of $v_k$, and that this process avoids backtracking.

\begin{proposition}
If $\S$ is a feasible 0--1 constraint set over $x$ and the branching order is $x_1,\ldots, x_n$, achieving sequential LP $k$-consistency for $k=1,\ldots, n$ suffices to solve $\S$ without backtracking.
\end{proposition}
\begin{proof}
Since $\S$ is feasible, $S_{\LP}$ is feasible at the root node of the branching tree, and so the empty assignment is consistent with $\S_{\LP}$.  \mbox{Arguing} inductively, suppose the partial assignment $(x_1,\ldots, x_{k-1})=(v_1,\ldots, v_{k-1})$ that reflects the branching decisions down to the node at level $k$ is consistent with $\S_{\LP}$. Since $\S$ is sequentially LP $k$-consistent, there exists a 0--1 value $v_k$ of $x_k$ for which the partial assignment $(x_1,\ldots,x_k)=(v_1,\ldots,v_k)$ is consistent with $\S_{\LP}$.  By induction, $\S_{\LP}\cup \{(x_1,\ldots,x_n)=(v_1,\ldots,v_n)\}$ is feasible at the terminal node of the tree for some tuple $(v_1,\ldots,v_n)$ of 0--1 values.  But in this case, $(x_1,\ldots,x_n)=(v_1,\ldots,v_n)$ satisfies $\S$, and we have solved the problem without backtracking.
\qed 
\end{proof}

\begin{example}  \label{ex:LPconsistency2}
Consider the constraint set $\S$ of Example~\ref{ex:LPconsistency}.  $\S$ is not sequentially LP 2-consistent because $x_1=0$ is consistent with $\S_{\LP}$, but neither $(x_1,x_2)=(0,0)$ nor $(x_1,x_2)=(0,1)$ is consistent with $\S_{\LP}$.  Also, backtracking is possible, because if we set $x_1=0$ at the root node because $x_1=0$ is consistent with $\S_{\LP}$, we cannot find a consistent value for $x_2$ at the child node and must backtrack.
Now suppose we add the clause $x_1+x_2\geq 1$ to $\S$ to obtain a constraint set $\S'$ that is sequentially LP 2-consistent (Fig.~\ref{fig:LPconsistency2}).  At the root node we must branch on $x_1=1$, because $x_1=0$ is not consistent with $\S'_{\LP}$.  At the child node, branching on $x_2=1$ yields an assignment $(x_1,x_2)=(1,1)$ that is consistent with $\S_{\LP}$ and, in fact, solves $\S$ without backtracking.
\end{example}

%%%%%%%%%%%%%%%%%%%%%%%%%%
%%%%%%%%%%%%%%%%%%%%%%%%%%%%%%%%%%%%%%%%%%%%%%
\section{Achieving LP Consistency} \label{LPconsistency}
%%%%%%%%%%%%%%%%%%%%%%%%%%%%%%%%%%%%%%%%%%%%%%
%%%%%%%%%%%%%%%%%%%%%%%%%%

The definition of sequential LP $k$-consistency, Definition~\ref{def:LPconsistency}, already suggests a method for achieving it.  Namely, we wish to obtain the results of lifting the problem from $k-1$ dimensions to $k$ dimensions, and then projecting back onto $k-1$ dimensions.  This operation can be achieved by using basics of the lift-and-project method of \cite{Bal85} as defined next.

To achieve sequential LP $k$-consistency, we proceed as follows.  Given $\S=\{Ax\geq b,\; x\in\{0,1\}^n\}$ where $0 \leq x_i \leq 1$ is included in $Ax \geq b$, we generate the nonlinear system
\begin{align*}
& (Ax-b) x_k \geq 0 \\
& (Ax-b) (1-x_k) \geq 0
\end{align*}
We next linearize the system by replacing each $x_k^2$ with $x_k$, and each product $x_i x_k$ with $y_{\{i,k\}}$.  Let the resulting system be $R_k(\S_{\LP})$.  Finally, project this system onto $x_J$ to obtain the system $R_k(\S_{\LP})|_{J_{k-1}}$.

\begin{proposition}
Given a 0--1 constraint set $\S$, applying the above algorithm and augmenting $\S$ with the constraints in $R_k(\S_{\LP})$ yields a constraint set that is sequentially LP $k$-consistent.
\end{proposition}

\begin{proof}
For a given 0--1 partial assignment $x_J=v_J$, suppose that $\S_{\LP}\cup\{(x_{J_k})=(v_{J_k})\}$ is infeasible for $v_k=0,1$.  It suffices to show that $R_k(\S_{\LP})|_{J_{k-1}}\cup\{x_{J_{k-1}}=v_{J_{k-1}}\}$ is infeasible.  It follows from \cite{Bal85} that $R_k(\S_{\LP})$ describes the convex hull of the union of $D(\S_{\LP}\cup\{x_k=v_k\})$ over $v_k=0,1$.  We claim that $x_{J_{k-1}}=v_{J_{k-1}}$ does not satisfy $R_k(\S_{\LP})|_{J_{k-1}}$.  Assume to the contrary.  Then there exists a point $w=(v_{J_{k-1}},\tilde{v}_k,\tilde{v}_K,\tilde{y})$ that satisfies $R_k(\S_{\LP})$, where $K=N\setminus J_k$.  This point must be representable as a convex combination of two points of the form $(v_{J_{k-1}},0,\dot{v}_K,\dot{y})$ and $(v_{J_{k-1}},1,\ddot{v}_K,\ddot{y})$, since the components of $v_{J_{k-1}}$ are integral and cannot be represented as the convex combination of other points.  However, by assumption such points do not exist because $\S_{\LP}\cup\{x_{J_k}=v_{J_k}\}$ is infeasible for $v_k=0,1$.  This yields the desired contradiction.  \qed
\end{proof}

\begin{example} \label{ex:RLT}
Consider again Example~\ref{ex:LPconsistency}, in which 
\[
\S = \{2x_1-4x_2 \leq -1, \; -2x_1+4x_2\leq 3, \; x_1,x_2\in \{0,1\}\}
\]
Recall that $\S$ is not LP 2-consistent because $x_1=0$ is consistent with $\S_{\LP}$ and $(x_1,x_2)=(0,v_2)$ is inconsistent with $\S_{\LP}$ for $v_2=0,1$.  We wish to achieve sequential LP 2-consistency by applying the modified lift-and-project procedure.  First generate the constraints
\[
\begin{array}{l@{\hspace{5ex}}l}
     (2x_1-4x_2+1)x_2 \leq 0 & x_1x_2\geq 0 \vspace{0.5ex} \\
     (2x_1-4x_2+1)(1-x_2) \leq 0 & x_1(1-x_2) \geq 0 \vspace{0.5ex} \\
     (-2x_1+4x_2-3)x_2 \leq 0 & (1-x_1)x_2 \geq 0 \vspace{0.5ex} \\
     (-2x_1+4x_2-3)(1-x_2) \leq 0 & (1-x_1)(1-x_2) \geq 0
\end{array}
\]
After linearizing and writing $y_{\{1,2\}}$ simply as $y$, we obtain the system $R_2(\S_{\LP})$:
\[
\begin{array}{l@{\hspace{5ex}}l}
     -3x_2 + 2y \leq 0 & y\geq 0 \vspace{0.5ex} \\
     2x_1 - x_2 -2y + 1 \leq 0 & x_1 - y \geq 0 \vspace{0.5ex} \\
     x_2 - 2y \leq 0 & x_2 - y \geq 0 \vspace{0.5ex} \\
     -2x_1 + 3x_2 + 2y - 3 \leq 0 & -x_1 - x_2 + y + 1 \geq 0
\end{array}
\]
Finally, projecting onto $x_1$ yields $R_2(\S)|_{\{1\}} = \{\frac{1}{2}\leq x_1\leq 1\}$.  Adding this constraint to $\S_{\LP}$, as illustrated in Fig.~\ref{fig:LPconsistency3}, achieves sequential LP 2-consistency because $x_1=0$ is inconsistent with the resulting constraint set.
\end{example}

\begin{figure}[!b]
\centering
\includegraphics[scale=0.6,clip,trim=240 140 290 150]{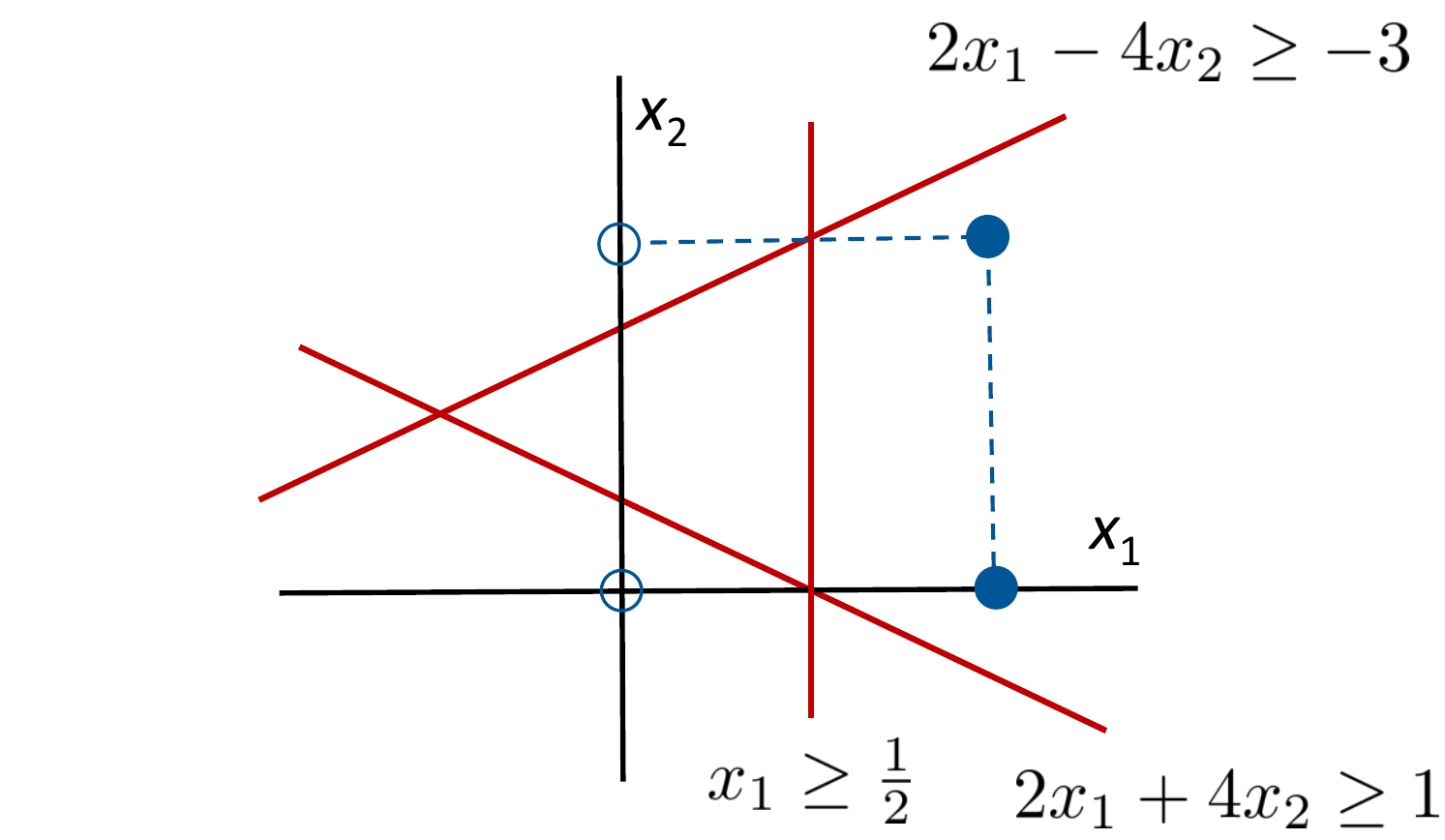}
\caption{Illustration of Example~\ref{ex:RLT}.}\label{fig:LPconsistency3}
\end{figure}

One could, in principle, apply lift-and-project repeatedly to achieve sequential LP $k$-consistency for $k=1,\ldots, n$, which would allow one to avoid backtracking altogether.  This is impractical, however, because the lift-and-project process quickly explodes in complexity as $k$ increases.  However, it may be practical to achieve sequential LP $k$-consistency for a few small $k$, a strategy that has three advantages.  First, it is computationally manageable for sufficiently small $k$.  Second, it generates sparse cuts (cuts with at most $k-1$ terms), which are generally conceived as the most effective type of cuts in branch-and-bound.

The third advantage is that achieving sequential LP-consistency can avoid branching that traditional cutting planes do not avoid, because it focuses on excluding inconsistent partial assignments, rather than on tightening the LP relaxation by cutting off fractional points.  This can be illustrated in a very simple context as follows.

\begin{example} \label{ex:tree}
Suppose we wish to maximize $3x_2-x_1$ subject to the constraint set $\S$ in the previous example.  Suppose further that we apply a traditional branch-and-cut procedure that generates separating disjunctive cuts at the root node (Fig.~\ref{fig:tree}(a)).  The solution of the LP relaxation at the root node is $(x_1,x_2)=(\frac{1}{2},1)$.  The two disjunctive cuts at this node are $-x_1+4x_2\geq -3$ (corresponding to the disjunction $x_1=0 \vee x_1=1$) and $x_1\geq \frac{1}{2}$ (corresponding to $x_2=0 \vee x_2=1$).  Only the first cut is generated, because only it cuts off the fractional solution $(\frac{1}{2},1)$.  This results in a new LP solution $(x_1,x_2)=(0,\frac{3}{4})$.  The procedure then branches on $x_2$.  The $x_2=0$ branch yields the fractional LP solution $(x_1,x_2)=(\frac{1}{2},0)$, and it is necessary to branch on $x_1$.  The $x_2=1$ branch yields the integer LP solution $(x_1,x_2)=(1,1)$, which solves the problem.  The resulting search tree has 5 nodes.

Suppose now that we achieve sequential LP 2-consistency as described in Example~\ref{ex:RLT} by generating the inequality $x_1\geq \frac{1}{2}$, even though it does not cut off the fractional LP solution (Fig.~\ref{fig:tree}(b)).  Since the partial assignment $x_1=0$ is inconsistent with the LP relaxation, we immediately branch on $x_1=1$, which yields the integer LP solution $(x_1,x_2)=(1,1)$.  The problem is solved with only 2 nodes in the search tree, even though we used no traditional separating cuts at all.
\end{example}

\begin{figure}
\centering
\includegraphics[scale=0.5,clip,trim=140 140 240 110]{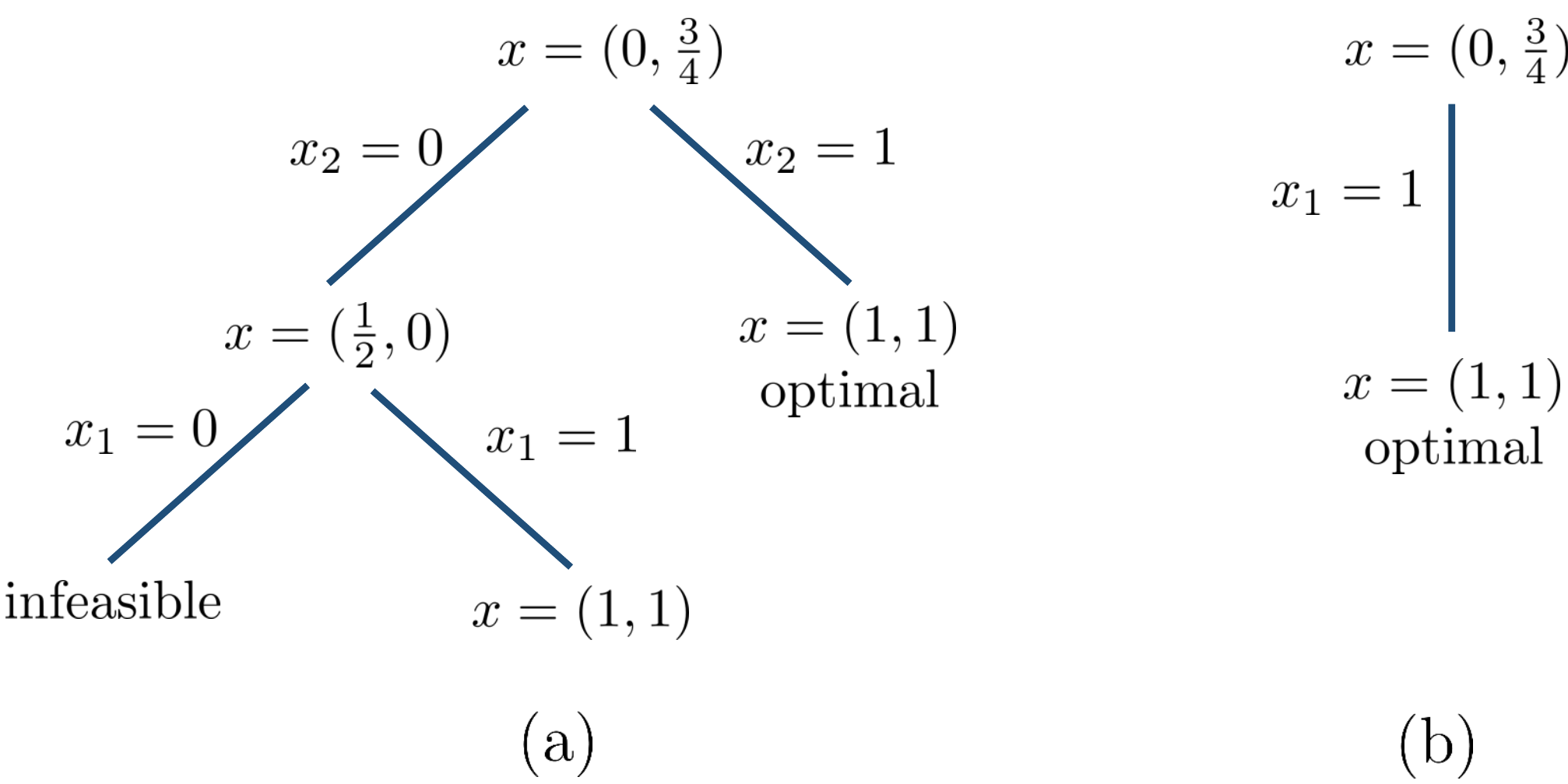}
\caption{Illustration of Example~\ref{ex:tree}}.\label{fig:tree}
\end{figure}

\section{Conclusion}

We provided a theoretical foundation for a new type of consistency, LP-consistency,  that is particularly suited to 0--1 programming.  It is based on the idea that consistency can, in general, be defined with respect to a type of relaxation.  LP-consistency is obtained by replacing the relaxation used for traditional consistency concepts with the LP relaxation.   We also defined sequential LP $k$-consistency, a weaker form of LP-consistency that is easier to achieve but nonetheless reduces backtracking.  In fact, sequential $k$-consistency can be obtained by a restricted form of the well-known lift-and-project process of integer programming. LP-consistency maintenance brings a new approach to 0--1 programming because it focuses on eliminating inconsistent 0--1 partial assignments rather than fractional solutions of the LP relaxation.  We showed that achieving even sequential \mbox{2-consistency} can avoid backtracking that traditional cutting planes allow. 

This work points to at least three further research programs.  One is to extend the concepts introduced here to general mixed integer/linear programming (MILP), which appears to be straightfoward.  A second is to investigate the computational usefulness of sequential LP $k$-consistency for MILP solvers, in particular by achieving sequential LP $k$-consistency for small $k$ near the top of the search tree.  A third is to conduct a systematic study of the ability of traditional cutting planes to achieve consistency, both traditional forms and LP-consistency, in an MILP problem.  This could allow one to make better use of known cutting planes by generating cuts that do not separate fractional solutions but enhance the consistency properties of the constraint set.

%%%%%%%%%%%%%%%%%%%%%%%%%%%%%%
%%%%%%%%%%%%%%%%%%%%%%%%%%%%%%%%%%%%%%%%%%%%%%%%%%%%%%%%%%%%
%%%%%%%%%%%%%%%%%%%%%%%%%%%%%%%%%%%%%%%%%%%%%%%%%%%%%%%%%%%%%%%%%%%%%%%%%%%%%%%%%%%%%%%%%%
% References
%%%%%%%%%%%%%%%%%%%%%%%%%%%%%%%%%%%%%%%%%%%%%%%%%%%%%%%%%%%%%%%%%%%%%%%%%%%%%%%%%%%%%%%%%%
%%%%%%%%%%%%%%%%%%%%%%%%%%%%%%%%%%%%%%%%%%%%%%%%%%%%%%%%%%%%
%%%%%%%%%%%%%%%%%%%%%%%%%%%%%%

%\bibliographystyle{plain}
%\bibliographystyle{spbasic}
%\bibliographystyle{splncs04}
%\bibliography{jnh}

\begin{thebibliography}{10}
\providecommand{\url}[1]{\texttt{#1}}
\providecommand{\urlprefix}{URL }
\providecommand{\doi}[1]{https://doi.org/#1}

\bibitem{Apt03}
Apt, K.R.: Principles of Constraint Programming. Cambridge University Press,
  Cambridge, UK (2003)

\bibitem{Bal85}
Balas, E.: Disjunctive programming and a hierarchy of relaxations for discrete
  optimization problems. SIAM Journal on Algebraic and Discrete Methods
  \textbf{6},  466--485 (1985)

\bibitem{Cha70}
Chang, C.L.: The unit proof and the input proof in theorem proving. Journal of
  the ACM  \textbf{14},  698--707 (1970)

\bibitem{Chv73}
Chv\'{a}tal, V.: Edmonds polytopes and a hierarchy of combinatorial problems.
  Discrete Mathematics  \textbf{4},  305--337 (1973)

\bibitem{Dav87}
Davis, E.: Constraint propagation with intervals labels. Artificial
  Intelligence  \textbf{32},  281--331 (1987)

\bibitem{Fre78}
Freuder, E.C.: Synthesizing constraint expressions. Communications of the ACM
  \textbf{21},  958--966 (1978)

\bibitem{Fre82}
Freuder, E.C.: A sufficient condition for backtrack-free search. Communications
  of the ACM  \textbf{29},  24--32 (1982)

\bibitem{Hoo89}
Hooker, J.N.: Input proofs and rank one cutting planes. ORSA Journal on
  Computing  \textbf{1},  137--145 (1989)

\bibitem{Hooker12}
Hooker, J.N.: Integrated Methods for Optimization, 2nd ed. Springer (2012)

\bibitem{Hoo16}
Hooker, J.N.: Projection, consistency, and {George Boole}. Constraints
  \textbf{21},  59--76 (2016)

\bibitem{Mac77}
Mackworth, A.: Consistency in networks of relations. Artificial Intelligence
  \textbf{8},  99--118 (1977)

\bibitem{Mon74}
Montanari, U.: Networks of constraints: Fundamental properties and applications
  to picture processing. Information Science  \textbf{7},  95--132 (1974)

\bibitem{Qui52}
Quine, W.V.: The problem of simplifying truth functions. American Mathematical
  Monthly  \textbf{59},  521--531 (1952)

\bibitem{Qui55}
Quine, W.V.: A way to simplify truth functions. American Mathematical Monthly
  \textbf{62},  627--631 (1955)

\bibitem{Reg10}
{R\'e}gin, J.C.: Global constraints: A survey. In: Milano, M., {Van
  Hentenryck}, P. (eds.) Hybrid Optimization: The Ten Years of CPAIOR, pp.
  63--134. Springer, New York (2010)

\bibitem{Tsa93}
Tsang, E.: Foundations of Constraint Satisfaction. Academic Press, London
  (1983)

\bibitem{Hen89}
{Van Hentenryck}, P.: Constraint Satisfaction in Logic Programming. {MIT}
  Press, Cambridge, {MA} (1989)

\end{thebibliography}

\end{document}